\documentclass[runningheads]{llncs}

\usepackage[T1]{fontenc}
\usepackage[utf8]{inputenc}
\usepackage{amsmath,amssymb,mathtools}
\usepackage{algorithm}
\usepackage{algpseudocode}
\usepackage{booktabs}
\usepackage{graphicx}
\usepackage{multirow}
\usepackage{enumitem}
\usepackage{url}
\usepackage{placeins}
\usepackage{xcolor}
\usepackage{tikz}
\usepackage[hidelinks]{hyperref}

\graphicspath{{figures/}}
\newcommand{\Reg}{\ifmmode\operatorname{Regular}\else\textnormal{\textsc{Regular}}\fi}
\newcommand{\DFA}{\mathcal{A}}
\newcommand{\States}{Q}
\newcommand{\Alphabet}{\Sigma}
\newcommand{\Trans}{\delta}
\newcommand{\Init}{q_0}
\newcommand{\Acc}{F}
\newcommand{\calL}{\mathcal{L}}

\newcommand{\wgt}{w}
\newcommand{\img}{\mathsf{Img}}
\newcommand{\INF}{+\infty}

\definecolor{orcidgreen}{HTML}{A6CE39}
\newcommand{\orcidbadge}[1]{%
  \href{https://orcid.org/#1}{%
    \raisebox{0.35ex}{%
      \tikz[baseline=-0.55ex]{%
        \node[circle,fill=orcidgreen,inner sep=0pt,minimum size=1.65ex]
          {\color{white}\sffamily\fontsize{4.2}{4.2}\selectfont iD};}}}}
\renewcommand{\orcidID}[1]{\unskip\,\orcidbadge{#1}}
\newcommand{\equalcontrib}{\textsuperscript{\(\dagger\)}}
\newcommand{\correspondingauthor}{\textsuperscript{*}}

\newcommand{\ExperimentSeed}{20260508}
\newcommand{\QualityInstances}{250}
\newcommand{\CompilerInfo}{C++17, g++ -O3}

\newcommand{\ValidationChecked}{60}
\newcommand{\ValidationFailed}{0}

\begin{document}

\title{Label Cuts in Layered Automata: Short Explanations for the Regular Constraint}
\titlerunning{Label-Cut Explanations for \Reg}

\author{XinYi Zhu\orcidID{0009-0009-0334-0194}\inst{1}\equalcontrib\and
Zonglin Yang\orcidID{0009-0000-5732-5699}\inst{2}\equalcontrib\correspondingauthor}
\authorrunning{X. Zhu and Z. Yang}
\institute{Computer Science and Engineering, Southern University of Science and Technology, China\\
\email{12412254@mail.sustech.edu.cn}
\and
Criminal Science and Technology, Guangdong Police College, China\\
\email{3258244847@qq.com}\\[0.35em]
\textsuperscript{\(\dagger\)}\textit{These authors contributed equally.}\quad
\textsuperscript{*}\textit{Corresponding author.}}

\maketitle

\begin{abstract}
Lazy clause generation solvers learn from propagation explanations; the form of a \Reg\ explanation therefore determines which clauses reach the learning engine.
Standard decompositions introduce variables for automaton states and derive reasons from local table constraints. BDD and MDD propagators instead select individual diagram edges.
We study the direct LCG explanation language for \Reg: disequality literals of the form $X_i\neq v$.
Here, one literal removes every automaton transition with the same position/value label, so an explanation is not an ordinary edge cut.
We prove that explanations for value deletions are exactly label $s$--$t$ cuts in the layered DFA unfolding.
This characterization establishes NP-completeness for minimum-weight explanations, gives a polynomial algorithm for inclusion-minimal explanations, and leads to an exact $O((U+n)4^{|\States|})$ dynamic program for small automata.
A C++ benchmark compares label-cut explanations with Table-LCG and MDD-edge baselines on \QualityInstances{} generated instances.
On average, minimal label cuts reduce projected explanation size from 22.2 to 5.5 literals and run in tens of microseconds.
Minimal label cuts are therefore the practical default; the exact algorithm provides an oracle for small automata.
\end{abstract}

\section{Introduction}

Lazy clause generation (LCG) combines propagation over finite domains with conflict-driven clause learning.
Whenever a propagator removes a value, it must provide an explanation: a set of already true literals that imply the removal.
The solver then records the implication as a clause and uses it for propagation and backjumping \cite{ohrimenko2007propagation,feydy2009lazy,katsirelos2005generalized,bacchus2007gac}.
Reason length is not the only concern in every solver state, but shorter explanations reduce the immediate reason passed to the learning engine.

The \Reg\ constraint is a standard way to model sequences accepted by a deterministic finite automaton (DFA) \cite{pesant2004regular,rossi2006handbook}.
Domain consistency filtering is well understood: a layered unfolding of the automaton identifies which pairs of positions and values participate in an accepting run.
Explanations are subtler.
One common model decomposes \Reg\ into state variables and ternary table constraints over $(S_{i-1},X_i,S_i)$ \cite{ohrimenko2008modelling,nethercote2007minizinc}. Diagram propagators can instead select BDD or MDD edges \cite{subbarayan2008efficient,gange2011mdd}.
Those explanation languages are local to states or edges.
This paper instead asks what can be achieved in the direct domain-literal language used by a \Reg\ propagator: literals $(X_i\neq v)$.

Domain literals couple the costs across automaton states.
If a literal $(X_i\neq v)$ is used in an explanation, every transition at layer $i$ labeled by value $v$ disappears at once.
Thus a \Reg\ explanation is not an edge cut in the layered graph; it is a cut over labels, where deleting one label can remove many edges.
This cost sharing makes the optimization difficult, but it can also produce short explanations.

We define label $s$--$t$ cuts for the conditioned layered DFA and prove that they coincide with LCG explanations for value deletions.
We then prove that finding a minimum-weight explanation is NP-complete even with unit weights and fixed sequence length, whereas a greedy deletion check computes an inclusion-minimal explanation in polynomial time.
For small automata, we give an exact dynamic program with running time $O((U+n)4^{|\States|})$, where $U$ is the number of removed domain values.
A C++ implementation compares these algorithms with Table-LCG, MDD-edge, reachability-greedy, and post-minimized baselines, as well as assignment propagation.

The experiments isolate explanation generation in a C++ benchmark; this version does not include integration with Chuffed.
We compare directly with the standard decomposition and diagram explanation languages.
Across \QualityInstances{} explanation instances generated with seed \ExperimentSeed, minimal label cuts reduce projected domain-literal explanations by roughly 75\% relative to Table-LCG, and the C++ extractor runs in microseconds to low milliseconds.
The exact algorithm provides an oracle for small automata. Minimal label cuts remain the practical default.

\section{Background and Related Work}

\paragraph{The \Reg\ constraint.}
Let $\DFA=(\States,\Alphabet,\Trans,\Init,\Acc)$ be a DFA, where $\Trans:\States\times\Alphabet\to\States$ may be partial.
The constraint $\Reg(X_1,\ldots,X_n,\DFA)$ holds when the word $x_1\cdots x_n$ labels an accepting run from $\Init$ to a state in $\Acc$.
Pesant's filtering algorithm builds a layered graph with vertices $(i,q)$ for layer $i$ and automaton state $q$; forward and backward reachability identify all supported pairs of positions and values \cite{pesant2004regular}.
We reuse this unfolding to determine which previous deletions explain an unsupported value.

\paragraph{Lazy clause generation explanations.}
LCG solvers combine CP propagation with SAT learning \cite{ohrimenko2007propagation,feydy2009lazy}.
For a derived deletion $(X_k\neq a)$, an explanation is a set of already true literals $E$ such that the constraint and $E$ entail the deletion.
We restrict attention to disequality literals $(X_i\neq v)$, not because LCG contains only such literals, but because this is the direct domain explanation language for the \Reg\ propagator.
Weights model solver preferences such as activity or clause management cost; unit weights recover shortest explanations.

\paragraph{Clause and proof minimization context.}
CDCL solvers also simplify learned clauses through implication graph analysis, recursive literal minimization, and clause quality heuristics \cite{marques1999grasp,een2003extensible,audemard2009predicting}.
Those techniques operate after a propagator has emitted a reason.
Our problem is the complementary one: generate a structurally correct reason inside the \Reg\ propagator before clause minimization at the solver level sees it.
The distinction also applies to hybrid optimization systems in which CP, SAT, and integer programming components exchange explanations in different native languages \cite{hooker2011logic}.

\paragraph{Table decomposition.}
A common decomposition introduces state variables $S_i$ and ternary table constraints over consecutive triples.
This gives a fixed modeling cost and reuses table explanations.
The explanation language, however, contains local state/table literals.
When projected back to original domain literals, the resulting reasons can include many values that are locally relevant somewhere but not globally necessary.
Our experiments therefore report both raw Table-LCG reason atoms and their projection to domain literals.

\paragraph{BDD and MDD explanations.}
BDD and MDD propagators naturally describe explanations in terms of edges.
Subbarayan studied nogood reasoning for BDDs \cite{subbarayan2008efficient}, and Gange et al. developed MDD propagators with explanations, including minimal explanation maintenance \cite{gange2011mdd}.
MDDs are closely related to layered automata and have been used for compact constraint propagation \cite{hoda2010systematic,perez2014improving}.
The two methods optimize different objects.
An MDD edge explanation removes one diagram arc for each selected edge.
In a \Reg\ domain-literal explanation, selecting $(X_i\neq v)$ removes all layer-$i$ arcs labeled $v$.
The label-cut problem studied here is therefore not the same as an edge-cut problem, even though both live on layered DAGs.

\paragraph{Label and colored cuts.}
General minimum label $s$--$t$ cut, also called colored cut, asks for a minimum set of edge labels whose deletion disconnects $s$ and $t$.
It is a known NP-hard group deletion problem with approximation, hardness, and parameterized results \cite{zhang2011approximation,zhang2016label,morawietz2022refined}.
Our hardness theorem is therefore not a claim that label cut is new.
It shows that hardness persists under the structure imposed by \Reg: a layered acyclic graph with vertices $(i,q)$, labels confined to one layer and value, and target or current-domain edges that cannot be removed.
Existing label-cut algorithms parameterized by path length, label frequency, or structural graph distance do not yield the exact $|\States|$-parameterized dynamic program below. In \Reg\ instances, path length is tied to the sequence length $n$, whereas the useful small parameter is the number of DFA states.
Equivalently, the problem is a hitting set over all accepting $s$--$t$ paths, represented implicitly by the DFA.

\paragraph{Global constraints defined by automata.}
Automata and grammar constraints are a standard CP tool for sequence modeling \cite{pesant2004regular,quimper2006global,hopcroft2001introduction}.
We add an explanation layer to an existing \Reg\ filtering event rather than a new filtering algorithm.
The algorithms reuse the layered graph already maintained by the filter.

\section{Explanations as Label Cuts}

At a propagation state, each variable has an initial domain $D_i^0$ and a current domain $D_i\subseteq D_i^0$.
Let $U_i=D_i^0\setminus D_i$ be the values removed so far.
For a target deletion $(X_k\neq a)$, the allowed explanation labels are
\[
  \calL=\{(i,v): i\in\{1,\ldots,n\},\ v\in U_i\}.
\]
An explanation is a subset $E\subseteq\calL$ satisfying
\[
\Reg(X_1,\ldots,X_n,\DFA)\wedge \bigwedge_{(i,v)\in E}(X_i\neq v)
  \models (X_k\neq a).
\]
The label at the target position is not needed: to explain why $a$ is removed from $X_k$, we reason about hypothetical supports with $X_k=a$.

\begin{definition}[Conditioned layered unfolding]
The conditioned graph $G_{k,a}$ has vertices $(i,q)$ for $i=0,\ldots,n$ and $q\in\States$, plus a sink $t$.
Its source is $s=(0,\Init)$.
For every layer $i\neq k$, state $q$, and value $v\in D_i^0$, the graph contains an edge $((i-1,q),(i,\Trans(q,v)))$ labeled $(i,v)$ whenever the transition is defined.
At layer $k$, only edges labeled $(k,a)$ are kept.
Every accepting vertex $(n,f)$, $f\in\Acc$, has an edge to $t$ with a special label that cannot be removed.
\end{definition}

\begin{definition}[Label $s$--$t$ cut]
A set $E\subseteq\calL$ is a label $s$--$t$ cut in $G_{k,a}$ if deleting every edge whose label belongs to $E$ makes $t$ unreachable from $s$.
Deleting a label removes all edges labeled by the same position and value.
\end{definition}

\begin{theorem}[Label-cut equivalence]
\label{thm:labelcut}
For every $E\subseteq\calL$, $E$ is an explanation for $(X_k\neq a)$ if and only if $E$ is a label $s$--$t$ cut in $G_{k,a}$.
\end{theorem}

\begin{proof}
Every $s$--$t$ path in $G_{k,a}$ is an accepting DFA run with $X_k=a$, and every accepting run with $X_k=a$ traces such a path.
If $t$ is reachable after deleting labels in $E$, the corresponding assignment satisfies \Reg\ and all literals in $E$ while falsifying $(X_k\neq a)$; hence $E$ is not an explanation.
Conversely, if $E$ is not an explanation, there is an assignment satisfying \Reg\ and all literals in $E$ with $X_k=a$.
Its accepting run survives in $G_{k,a}\setminus E$, so $t$ is reachable.
\end{proof}

The theorem also separates the \Reg\ explanation problem from edge-cut results and explanations produced by table decompositions.
An edge cut pays separately for each edge; a Table-LCG reason uses local state/table atoms; a direct domain-literal reason pays once for every equal label at a layer.
Figure~\ref{fig:layered-unfolding} compares these three explanation languages and shows the modeling pitfall that would arise if one tried to merge literal nodes across states.
The phantom path panel warns against an incorrect reduction; it makes no claim about prior MDD implementations.

\begin{figure}[!htbp]
  \centering
  \includegraphics[width=\textwidth]{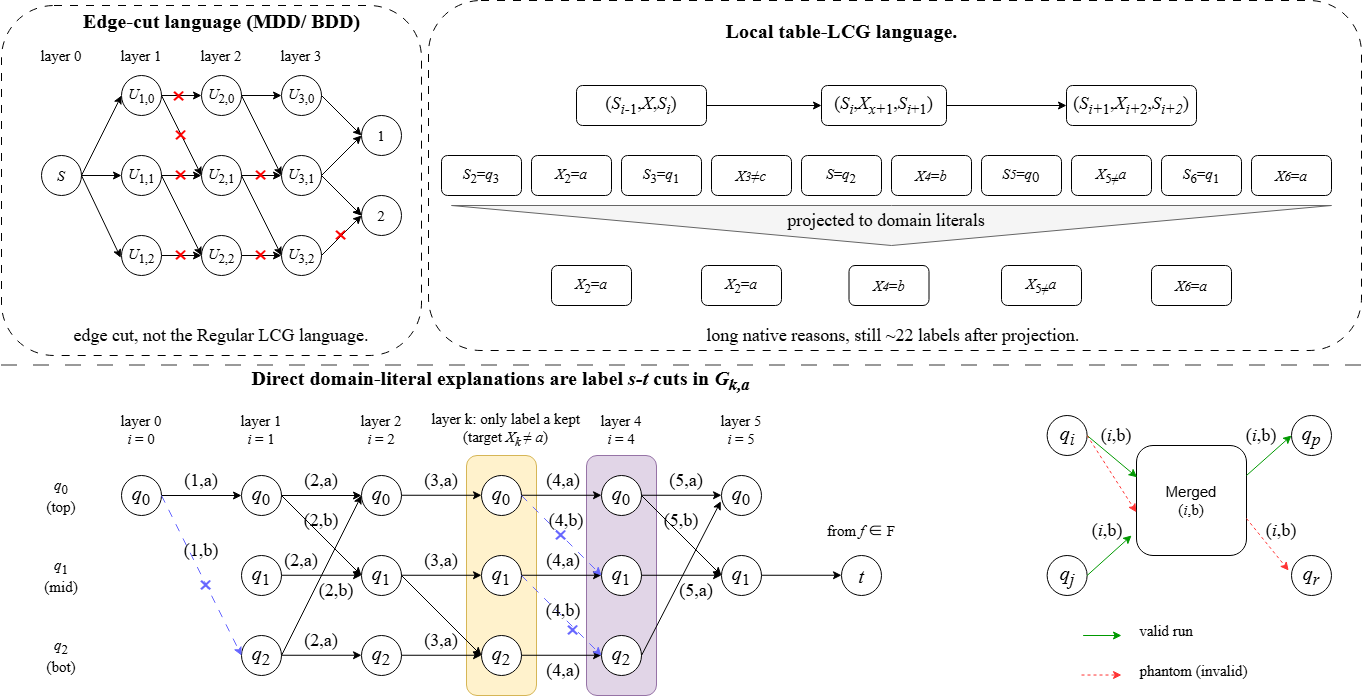}
  \caption{Three explanation languages for \Reg. Top left: MDD/BDD edge cuts pay per arc. Top right: Table-LCG local state reasons remain long, averaging about 22 labels in our benchmark even after projection. Bottom left: direct domain-literal explanations are label $s$--$t$ cuts in $G_{k,a}$; the literal $(4,b)$ removes every transition at layer 4 with that label. Bottom right: merging literal nodes across states can create invalid phantom paths, so Theorem~\ref{thm:labelcut} applies to labeled edges rather than collapsed literal nodes.}
  \label{fig:layered-unfolding}
\end{figure}

\section{Algorithms}

\subsection{Minimum and minimal explanations}

A minimum-weight explanation minimizes $\sum_{(i,v)\in E}\wgt_{i,v}$ over all label cuts.
Computing one is NP-hard.
The theorem below establishes hardness for direct \Reg\ explanations under restricted domains, complementing the general label-cut hardness results cited above.

\begin{theorem}
\label{thm:np}
Given a \Reg\ deletion, deciding whether an explanation of weight at most $B$ exists is NP-complete.
Hardness holds even with unit weights and $n=4$.
\end{theorem}

\begin{proof}[Proof sketch]
Membership in NP follows from Theorem~\ref{thm:labelcut}.
For hardness, reduce from vertex cover in 3-uniform 3-partite hypergraphs \cite{guruswami2015inapproximability}.
Given hyperedges $\mathcal{H}\subseteq V_1\times V_2\times V_3$, build an automaton over words of length four that accepts exactly $cuvw$ when $(u,v,w)\in\mathcal{H}$.
Allowed explanation labels correspond to vertices in the three parts.
Choosing labels that block every accepting word is exactly choosing a vertex cover.
Appendix~\ref{app:np-proof} gives the full construction, including $D_i^0$, $D_i$, the target deletion, unit weights, and a proof in both directions that preserves the budget.
\end{proof}

An inclusion-minimal explanation can instead be computed greedily.
Start with all allowed labels except those at the target position.
Then scan labels in descending weight order and remove a label whenever the remaining labels still disconnect $s$ and $t$.
Each check is a graph reachability query in $G_{k,a}$.
The result is inclusion-minimal, because no remaining label can be removed without reconnecting the sink.
The scan order selects one of the possible inclusion-minimal explanations; minimality itself does not depend on descending weight order.

\subsection{Exact dynamic program for small automata}

The exact algorithm uses reachable state subsets.
For each layer $i$, define mandatory values
\[
A_i=\begin{cases}
\{a\}, & i=k,\\
D_i, & i\neq k,
\end{cases}
\quad
U_i=\begin{cases}
\emptyset, & i=k,\\
D_i^0\setminus D_i, & i\neq k.
\end{cases}
\]
For $S\subseteq\States$ and value $v$, let
$\img_S(v)=\{\Trans(q,v):q\in S,\ \Trans(q,v)\text{ defined}\}$.
The DP state $\mathrm{DP}[i,S]$ stores the minimum deletion weight over the first $i$ layers that leaves exactly the reachable state set $S$.
The implementation precomputes all images $\img_S(v)$ for $S\subseteq\States$ and $v\in\Alphabet$ using bitsets in $O(|\Alphabet|2^{|\States|}|\States|)$ time.

For a transition from $S$ to a target set $T$, all mandatory images must be included:
$M=\bigcup_{v\in A_i}\img_S(v)\subseteq T$.
Because weights are positive, the optimal set of optional values retained for $T$ is
\[
K_i^\star(S,T)=\{v\in U_i:\img_S(v)\subseteq T\}.
\]
The transition is feasible iff
$M\cup\bigcup_{v\in K_i^\star(S,T)}\img_S(v)=T$.
Its cost is the weight of deleted optional values,
$\sum_{v\in U_i\setminus K_i^\star(S,T)}\wgt_{i,v}$.
The final set must avoid accepting states.

\begin{algorithm}[t]
\caption{Exact minimum-weight label-cut explanation}
\label{alg:exact}
\small
\begin{algorithmic}[1]
\State $\mathrm{DP}[0,\{\Init\}]\gets 0$ and $\mathrm{DP}[0,S]\gets \INF$ for all other $S$
\For{$i=1$ to $n$}
  \ForAll{$S,T\subseteq\States$ with $\mathrm{DP}[i-1,S]<\INF$}
    \State $M\gets\bigcup_{v\in A_i}\img_S(v)$
    \State $K_i^\star(S,T)\gets\{v\in U_i:\img_S(v)\subseteq T\}$
    \If{$M\subseteq T$ and $M\cup\bigcup_{v\in K_i^\star(S,T)}\img_S(v)=T$}
      \State update $\mathrm{DP}[i,T]$ with cost $\sum_{v\in U_i\setminus K_i^\star(S,T)}\wgt_{i,v}$
    \EndIf
  \EndFor
\EndFor
\State return the cheapest $\mathrm{DP}[n,S]$ with $S\cap\Acc=\emptyset$ and reconstruct deleted labels
\end{algorithmic}
\end{algorithm}

\begin{theorem}
Algorithm~\ref{alg:exact} returns a minimum-weight explanation in $O((U+n)4^{|\States|})$ time after the image precomputation above, and $O(n2^{|\States|})$ space with backpointers, where $U=\sum_i|U_i|$.
\end{theorem}

\begin{proof}[Proof sketch]
For fixed $i$, $S$, and $T$, the set $K_i^\star(S,T)$ is the unique set of maximum total weight that can be kept without reaching outside $T$; all other optional values must be deleted.
Induction on $i$ gives the DP invariant that $\mathrm{DP}[i,S]$ is the optimum cost for producing reachable set $S$ after layer $i$.
The final condition $S\cap\Acc=\emptyset$ is exactly disconnection from the accepting sink.
There are $2^{|\States|}$ choices for $S$ and $2^{|\States|}$ choices for $T$ per layer, so the transition loop enumerates $4^{|\States|}$ pairs.
Scanning optional values contributes $\sum_i |U_i|4^{|\States|}=U4^{|\States|}$, while checks of mandatory images and feasibility contribute $n4^{|\States|}$ with precomputed bitset images.
A rolling array stores only $O(2^{|\States|})$ costs; reconstructing the explanation stores one predecessor and one retention decision per layer and subset, giving $O(n2^{|\States|})$ space.
\end{proof}

The implementation computes layer transitions on demand from the precomputed images. Materializing every $(S,T)$ choice would add an $O(U4^{|\States|})$ table.
Meet-in-the-middle methods, symmetry breaking, and dominance pruning remain open implementation choices. Without them, the exact algorithm is intended as an oracle for small automata.

\paragraph{Assignment propagation.}
If \Reg\ entails an assignment $X_k=a$, then it entails every deletion $(X_k\neq b)$ for $b\in D_k^0\setminus\{a\}$.
We explain the assignment by taking the union of the deletion explanations for all such $b$.
The union need not be minimum when equality is a primitive literal, but it is sound and matches our language of deletion literals.

\section{Experimental Evaluation}

\paragraph{Implementation and setup.}
The main experiments use a standalone C++17 implementation compiled with \CompilerInfo.
The random seed is \ExperimentSeed.
All reported instances are generated explanation events for which the reason that includes every eligible literal is valid.
The exact DP was validated against brute force on \ValidationChecked{} small instances, with \ValidationFailed{} failures.
We retain the earlier Python CDCL simulation only as supplementary evidence; every runtime claim in the main paper comes from the C++ implementation.
The benchmark measures explanation generation in isolation rather than full solver performance.
The supplementary artifact contains the C++ implementation and generator, together with the raw CSV files and plotting script used to reproduce the tables and figures.

\paragraph{Baselines.}
\textsc{Naive} uses every available domain literal except literals at the target position.
\textsc{Table-LCG} simulates a decomposition with state variables and ternary table constraints: we report raw local reason atoms and their projection to distinct domain literals.
\textsc{MDD-edge} computes a minimum edge cut in the conditioned layered graph, using individual edges as reason atoms; its projected label count is diagnostic, not its native explanation language.
\textsc{Reachability} uses the BFS heuristic and is reported as the reachability-greedy baseline.
\textsc{Minimal} is the greedy label-deletion algorithm, and \textsc{Optimal} is Algorithm~\ref{alg:exact}.
For post-minimization, the same oracle label-deletion check is applied to baseline outputs.

\begin{figure}[H]
  \centering
  \includegraphics[width=\textwidth]{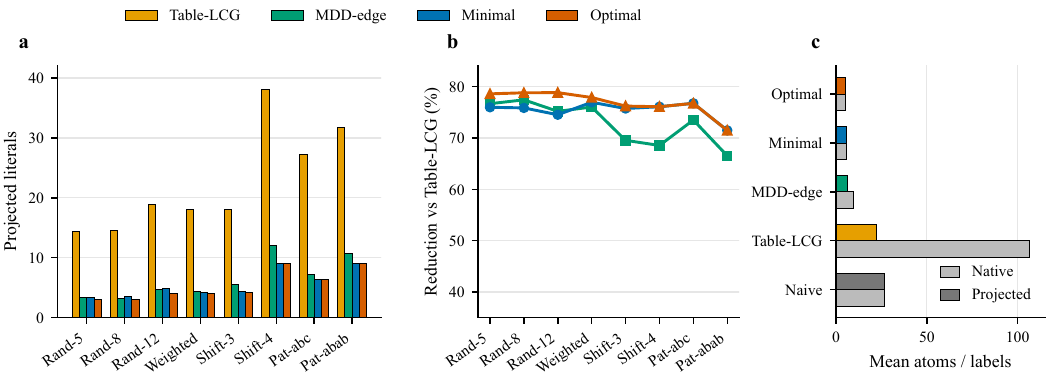}
  \caption{Main explanation quality results from the standalone C++ benchmark.
  (a) Projected domain-literal reason length by instance family.
  (b) Reduction relative to the projected Table-LCG baseline.
  (c) Native reason atoms versus projected domain literals, showing why Table-LCG and MDD-edge explanations must be reported in their own languages as well as after projection.}
  \label{fig:quality-overview}
\end{figure}

\begin{table}[H]
\centering
\caption{Diagnostic comparison of explanation languages, averaged over all C++ benchmark instances.}\label{tab:diagnostic}
\small
\begin{tabular}{lrrr}
\toprule
Method & Native atoms & Projected labels & Time ($\mu$s)\\
\midrule
Naive & 26.5 & 26.5 & 1.4\\
Table-LCG & 106.7 & 22.2 & 16.1\\
MDD-edge & 9.6 & 6.2 & 32.0\\
Reachability & 22.2 & 22.2 & 12.6\\
Minimal & 5.5 & 5.5 & 25.6\\
Optimal & 5.3 & 5.3 & 325.5\\
\bottomrule
\end{tabular}
\end{table}

Figure~\ref{fig:quality-overview} and Table~\ref{tab:diagnostic} separate the explanation languages.
Table-LCG has many raw local atoms and, after projection, averages 22.2 domain labels.
MDD-edge explanations are much shorter in their native edge language, but their projected labels remain a different object from a label-cut optimum.
Minimal label cuts average 5.5 labels, close to the 5.3-label exact optimum.

\begin{table}[!tbp]
\centering
\caption{C++ explanation quality. Native atoms count the reason atoms in each method's own language; projected labels count distinct domain literals $(X_i\neq v)$. Reduction is relative to Table-LCG projected labels.}\label{tab:quality-main}
\scriptsize
\setlength{\tabcolsep}{3.2pt}
\begin{tabular}{llrrrr}
\toprule
Family & Method & Native & Labels & Time ($\mu$s) & Red.\\
\midrule
Random-Q5 & Table-LCG & 64.1 & 14.3 & 9.0 & --\\
 & MDD-edge & 4.3 & 3.4 & 15.1 & 76.7\%\\
 & Reachability & 14.3 & 14.3 & 6.3 & 0.0\%\\
 & Minimal & 3.5 & 3.5 & 12.3 & 76.0\%\\
 & Optimal & 3.1 & 3.1 & 10.4 & 78.6\%\\
\midrule
Random-Q8 & Table-LCG & 98.1 & 14.5 & 16.2 & --\\
 & MDD-edge & 4.5 & 3.3 & 28.8 & 77.4\%\\
 & Reachability & 14.5 & 14.5 & 12.3 & 0.0\%\\
 & Minimal & 3.5 & 3.5 & 25.4 & 75.8\%\\
 & Optimal & 3.1 & 3.1 & 41.6 & 78.8\%\\
\midrule
Random-Q12 & Table-LCG & 183.9 & 18.9 & 23.6 & --\\
 & MDD-edge & 6.9 & 4.7 & 48.1 & 75.2\%\\
 & Reachability & 18.9 & 18.9 & 18.5 & 0.0\%\\
 & Minimal & 4.8 & 4.8 & 37.2 & 74.5\%\\
 & Optimal & 4.0 & 4.0 & 645.6 & 78.8\%\\
\midrule
Weighted-Q7 & Table-LCG & 89.0 & 18.1 & 14.5 & --\\
 & MDD-edge & 5.0 & 4.3 & 25.8 & 76.0\%\\
 & Reachability & 18.1 & 18.1 & 10.8 & 0.0\%\\
 & Minimal & 4.2 & 4.2 & 20.5 & 76.9\%\\
 & Optimal & 4.0 & 4.0 & 26.6 & 77.9\%\\
\midrule
Shift-3x2 & Table-LCG & 96.9 & 18.0 & 14.5 & --\\
 & MDD-edge & 9.1 & 5.5 & 25.7 & 69.5\%\\
 & Reachability & 18.0 & 18.0 & 10.9 & 0.0\%\\
 & Minimal & 4.4 & 4.4 & 20.0 & 75.7\%\\
 & Optimal & 4.3 & 4.3 & 24.3 & 76.2\%\\
\midrule
Shift-4x3 & Table-LCG & 257.2 & 38.1 & 36.2 & --\\
 & MDD-edge & 28.4 & 12.0 & 91.4 & 68.5\%\\
 & Reachability & 38.1 & 38.1 & 32.5 & 0.0\%\\
 & Minimal & 9.1 & 9.1 & 70.8 & 76.1\%\\
 & Optimal & 9.1 & 9.1 & 2781.3 & 76.1\%\\
\midrule
Pattern-abc & Table-LCG & 67.1 & 27.3 & 11.2 & --\\
 & MDD-edge & 9.1 & 7.2 & 19.2 & 73.5\%\\
 & Reachability & 27.3 & 27.3 & 8.3 & 0.0\%\\
 & Minimal & 6.3 & 6.3 & 16.8 & 76.7\%\\
 & Optimal & 6.3 & 6.3 & 8.1 & 76.7\%\\
\midrule
Pattern-abab & Table-LCG & 90.4 & 31.8 & 14.9 & --\\
 & MDD-edge & 14.3 & 10.7 & 31.1 & 66.5\%\\
 & Reachability & 31.8 & 31.8 & 11.3 & 0.0\%\\
 & Minimal & 9.1 & 9.1 & 23.0 & 71.5\%\\
 & Optimal & 9.1 & 9.1 & 14.4 & 71.5\%\\
\midrule
\bottomrule
\end{tabular}
\end{table}

Table~\ref{tab:quality-main} reports results by family.
Minimal label cuts reduce projected explanation size by 71.5--76.9\% relative to Table-LCG across all families.
The exact algorithm improves slightly on random and weighted instances and matches Minimal on the structured shift and pattern families.
Minimal is therefore the practical default; Optimal is better suited to small automata and offline analysis.

\begin{table}[!tbp]
\centering
\caption{Oracle post-minimization of baseline explanations. Total time includes producing the baseline reason and then applying the label-cut deletion check.}\label{tab:postmin}
\small
\begin{tabular}{lrrrrr}
\toprule
Method & Before & After & Gap & Total ($\mu$s) & Direct Minimal\\
\midrule
Naive+Min & 26.5 & 5.5 & 0.2 & 37.0 & 25.6\\
Table+Min & 22.2 & 5.5 & 0.2 & 48.6 & 25.6\\
Reachability+Min & 22.2 & 5.5 & 0.2 & 44.3 & 25.6\\
\bottomrule
\end{tabular}
\end{table}

Oracle post-minimization closes the gap: starting from Naive, Table-LCG, or Reachability and repeatedly deleting unnecessary labels reaches the same average size as Minimal in this benchmark.
This result identifies the source of the improvement. Strong minimization requires the label-cut reachability test, and Minimal applies that test directly instead of first constructing a long reason.
The oracle baseline is stronger than an ordinary clause minimizer at the solver level, which normally sees the implication graph rather than reachability in the conditioned DFA.
Table~\ref{tab:postmin} also reports total time: oracle post-minimization costs 37.0--48.6 microseconds after constructing the baseline reason, compared with 25.6 microseconds for direct Minimal.

\begin{figure}[t]
  \centering
  \includegraphics[width=\textwidth]{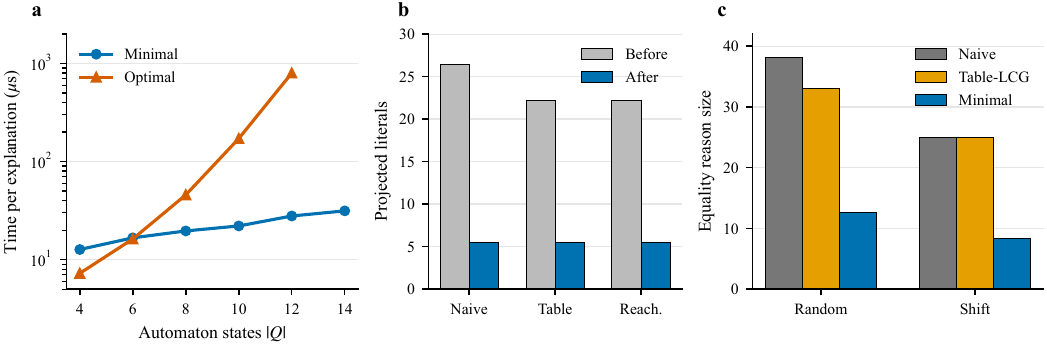}
  \caption{Runtime and extension results.
  (a) C++ runtime as a function of automaton size; Minimal stays in the microsecond range, while Optimal shows the expected exponential dependence on $|\States|$.
  (b) Oracle post-minimization reduces baseline reasons to the same size as Minimal, but first constructs and then minimizes a long reason.
  (c) Assignment propagation explanations, built by unioning deletion reasons for all excluded values.}
  \label{fig:runtime-extensions}
\end{figure}

\begin{table}[!htb]
\centering
\caption{C++ runtime scaling in microseconds per explanation. Optimal is reported only for automata sizes where the exact DP is intended to be used.}\label{tab:runtime}
\scriptsize
\setlength{\tabcolsep}{3.5pt}
\begin{tabular}{llrrrrrr}
\toprule
Axis & Quantity & \multicolumn{6}{c}{Samples}\\
\cmidrule(lr){3-8}
\midrule
$n$ & Setting & 8 & 16 & 32 & 64 & 128 & --\\
    & Minimal & 9.4 & 25.7 & 36.8 & 171.4 & 328.3 & --\\
    & Optimal & 11.4 & 43.9 & 20.0 & 133.2 & 141.9 & --\\
\midrule
$|D|$ & Setting & 3 & 4 & 5 & 6 & 8 & --\\
      & Minimal & 13.4 & 19.2 & 21.6 & 26.7 & 32.8 & --\\
      & Optimal & 16.5 & 33.4 & 44.2 & 22.6 & 19.4 & --\\
\midrule
$|Q|$ & Setting & 4 & 6 & 8 & 10 & 12 & 14\\
      & Minimal & 12.7 & 16.7 & 19.6 & 22.1 & 27.9 & 31.4\\
      & Optimal & 7.2 & 16.2 & 45.8 & 171.8 & 799.2 & --\\
\bottomrule
\end{tabular}
\end{table}

Figure~\ref{fig:runtime-extensions} and Table~\ref{tab:runtime} report C++ runtime scaling.
Minimal runs in a few microseconds to less than one millisecond across the tested values of $n$, $|\Alphabet|$, and $|\States|$.
Optimal grows with $|\States|$, as predicted by the $4^{|\States|}$ term, but remains useful for small automata and for measuring the optimality gap.
In this standalone implementation, the exact oracle is practical for $|\States|\le 14$; Minimal is the default beyond that range.
The non-monotonic rows for $n$ and $|\Alphabet|$ reflect that generated instances differ in the number of removed values and reachable states; the $|\States|$ trend is the main stress test for the exact DP.

\begin{figure}[t]
  \centering
  \includegraphics[width=\textwidth]{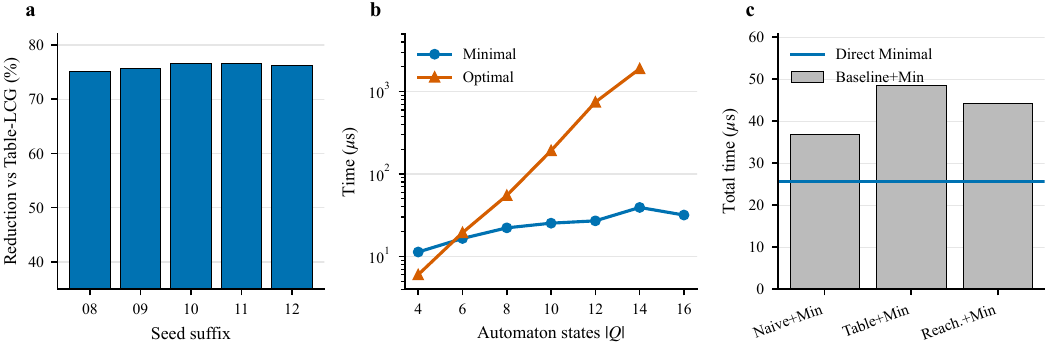}
  \caption{Supplementary stability and boundary checks.
  (a) Minimal's reduction over projected Table-LCG explanations across five seeds, each with 250 generated instances.
  (b) Runtime boundary for the exact oracle; $|Q|=16$ is shown only for Minimal.
  (c) Total time for oracle post-minimization compared with direct Minimal.}
  \label{fig:robustness-boundary}
\end{figure}

Across five seeds, Minimal reduces projected Table-LCG explanations by $76.0\pm0.6\%$ (Fig.~\ref{fig:robustness-boundary}).
The exact DP averages about 1.9 milliseconds at $|Q|=14$ on the boundary benchmark, while Minimal stays in the tens of microseconds.

\begin{table}[H]
\centering
\caption{Explanations for assignment propagation $\Reg\Rightarrow X_k=a$, constructed as the union of reasons for deleting all $b\neq a$.}\label{tab:equality}
\small
\begin{tabular}{llrr}
\toprule
Family & Method & Size & Time ($\mu$s)\\
\midrule
Random-Q7 & Naive & 38.1 & 41.4\\
 & Table-LCG & 33.1 & 90.9\\
 & Minimal & 12.6 & 141.7\\
\midrule
Shift-3x2 & Naive & 25.0 & 22.8\\
 & Table-LCG & 25.0 & 51.5\\
 & Minimal & 8.4 & 72.2\\
\midrule
\bottomrule
\end{tabular}
\end{table}

Table~\ref{tab:equality} reports assignment propagation.
When \Reg\ entails $X_k=a$, unioning deletion reasons for all $b\neq a$ still gives compact explanations.
Minimal reduces equality reasons from 33.1 to 12.6 labels on random instances and from 25.0 to 8.4 labels on shift instances.
Pattern instances rarely generated singleton assignment events under the fixed generator, so they are not reported in this table.

\section{Conclusion}

Direct domain-literal explanations for \Reg\ are label cuts.
Their semantics differ from both BDD/MDD edge explanations and local explanations produced by table decompositions with state variables.
The characterization proves NP-completeness for optimum explanations, yields a polynomial minimal algorithm, and supports an exact DP for small automata.

The C++ benchmark points to a straightforward implementation choice.
Minimal label-cut explanations are short, fast, and match post-minimized baselines without first producing a long reason.
Optimal explanations are best viewed as an oracle or a tool for small automata, not as the default production method.
Integration with full solver pipelines such as Chuffed, OR-Tools CP-SAT, or CP Optimizer remains the next engineering step; the present experiments isolate explanation generation and the standard competing languages.
Future work should evaluate standard CP models that use \Reg\ inside complete search loops and test DP pruning or meet-in-the-middle variants.

\section*{Conflict of Interest}

The authors declare that they have no conflicts of interest.

\FloatBarrier

\bibliographystyle{splncs04}
\bibliography{references}

\appendix
\section{Full Proof of Theorem~\ref{thm:np}}
\label{app:np-proof}

\begingroup
\small
\begin{proof}
\textbf{Membership in NP.}
Given $E\subseteq\calL$, delete from $G_{k,a}$ every edge labelled by a literal in $E$.
By Theorem~\ref{thm:labelcut}, $E$ is an explanation exactly when $t$ is unreachable from $s$.
Because $G_{k,a}$ has $O(n|\States|)$ vertices and $O(n|\States|\max_i|D_i^0|)$ edges, reachability and the budget test $\sum_{(i,v)\in E}\wgt_{i,v}\le B$ are polynomial.

\textbf{NP-hardness.}
We reduce from vertex cover in 3-uniform 3-partite hypergraphs, which is NP-hard~\cite{guruswami2015inapproximability}.
An instance has disjoint parts $V_1,V_2,V_3$, hyperedges $\mathcal{H}\subseteq V_1\times V_2\times V_3$, and a budget $B$.
Introduce fresh symbols $c,d_2,d_3,d_4$, let $\Alphabet=\{c\}\uplus V_1\uplus V_2\uplus V_3\uplus\{d_2,d_3,d_4\}$, and use four variables with
\[
\begin{array}{lll}
D_1^0=\{c\}, & D_2^0=V_1\cup\{d_2\}, &
D_3^0=V_2\cup\{d_3\}, \quad D_4^0=V_3\cup\{d_4\},\\[2mm]
D_1=\{c\}, & D_2=\{d_2\}, &
D_3=\{d_3\}, \quad D_4=\{d_4\}.
\end{array}
\]
Use the trie DFA for $\{cuvw:(u,v,w)\in\mathcal{H}\}$.
It is deterministic, has at most $1+4|\mathcal{H}|$ states, and never accepts a dummy symbol.
Thus the dummies keep the current domains nonempty, while the unit-weight explanation literals for the target deletion $(1,c)$ correspond exactly to the vertices in $V_1\cup V_2\cup V_3$.

\textbf{Forward direction.}
For a vertex cover $C$ with $|C|\le B$, define
\[
E_C=
\{(2,u):u\in C\cap V_1\}
\cup
\{(3,v):v\in C\cap V_2\}
\cup
\{(4,w):w\in C\cap V_3\}.
\]
Every accepted word $cuvw$ represents a hyperedge hit by $C$ and therefore violates a literal in $E_C$.
Hence $E_C$ explains $(X_1\neq c)$ and has weight $|C|\le B$.

\textbf{Reverse direction.}
For an explanation $E$ of weight at most $B$, set
\[
C_E=
\{u\in V_1:(2,u)\in E\}
\cup
\{v\in V_2:(3,v)\in E\}
\cup
\{w\in V_3:(4,w)\in E\}.
\]
If $C_E$ missed a hyperedge $(u,v,w)$, the accepted word $cuvw$ would satisfy every literal in $E$ while keeping $X_1=c$, a contradiction.
Thus $C_E$ is a vertex cover and $|C_E|\le|E|\le B$.
The polynomial reduction uses unit weights and fixed length $n=4$; together with membership in NP, this proves NP-completeness.
\end{proof}
\endgroup

\end{document}